\documentclass[preprint,12pt]{elsarticle}
\usepackage{etoolbox}
\makeatletter
\def\ps@pprintTitle{%
	\let\@oddhead\@empty
	\let\@evenhead\@empty
	\def\@oddfoot{\reset@font\hfil\thepage\hfil}
	\let\@evenfoot\@oddfoot
}
\makeatother
\usepackage{graphics}
\usepackage{amsmath, amsthm, amssymb}
\usepackage{natbib}
\biboptions{comma,round,authoryear}
\usepackage[colorlinks=true]{hyperref}
\usepackage{lscape}
\usepackage[utf8]{inputenc} 
\usepackage{amsthm}
\usepackage{enumerate}
\usepackage{mathrsfs}
\usepackage{setspace}
\usepackage{multirow}
\usepackage{graphicx}
\usepackage{amsfonts}
\usepackage{verbatim}
\usepackage{mathrsfs}
\usepackage{hyperref}
\usepackage{amssymb}
\usepackage{amsthm}
\usepackage{caption}
\usepackage{subcaption}
\usepackage{multirow,bigstrut,bigdelim}
\usepackage{listings}
\usepackage{xcolor} 
\newtheorem{theorem}{Theorem}[section]

\newtheorem{corollary}{Corollary}[section]
\theoremstyle{plain}
\newtheorem{definition}{Definition}[section]
\newtheorem{example}{Example}[section]

\theoremstyle{remark}

\numberwithin{equation}{section}
\usepackage[mathscr]{euscript}
\usepackage{geometry} 
\usepackage{graphicx,lscape} 
\usepackage{booktabs}
\usepackage{array}
\usepackage{paralist} 
\usepackage{verbatim}
\usepackage{makecell}
\usepackage{orcidlink}

\begin{document}
\begin{frontmatter}		
\title{\textbf{Extropy-Based Generalized Divergence and Similarity Ratios: Theory and Applications}}

\author{Saranya. P \orcidlink{0009-0002-2629-0248}\corref{cor1}}
		\ead{saranyapanat96@gmail.com}
         \author{S. M. Sunoj \orcidlink{0000-0002-6227-1506}} 
         \ead{smsunoj@cusat.ac.in}
		\cortext[cor1]{Corresponding author}
		\address{Department of Statistics\\Cochin University of Science and Technology\\Cochin 682 022, Kerala, INDIA.}
\begin{abstract}
    In this article, we propose two classes of relative information measures based on extropy, viz., the generalized extropy similarity ratio (GESR) and generalized extropy divergence ratio (GEDR), that measure the similarity and discrepancy between two probability distributions, respectively. Definitions of GESR and GEDR are proposed along with their fundamental axioms, properties, and some measures satisfying those axioms are also introduced. The relationship of GESR with the popular cosine similarity is also established in the study. Various properties of GESR and GEDR, including bounds under the proportional hazards model and the proportional reversed hazards model, are derived. Nonparametric estimators of GESR are defined, and their performance is evaluated using simulation studies. Applications of the GESR in lifetime data analysis and image analysis are also demonstrated in this study.     
\end{abstract}
\end{frontmatter}

\section{Introduction}
In information theory, some of the divergence measures, whether expressed through the probability density function (PDF) $f(\cdot)$, the cumulative distribution function (CDF) $F (\cdot)$, or the survival function (SF) $\bar{F}(\cdot) = 1 - F(\cdot)$ of two random variables $X$ and $Y$ can be represented as the difference between the corresponding inaccuracy measure and the average uncertainty measure. 
We can define a class of divergence measures in information theory as follows: 
Let $\phi_1(\cdot)$ and $\phi_2(\cdot)$ represent general probability measures corresponding to $X$ and $Y$ respectively such that $\phi_i (\cdot)$ can be a probability density function, distribution function, or survival function. Let us denote the average uncertainty contained in the probability measure $\phi_1 (\cdot)$ about the predictability of an outcome of the random variable $X$ as $\mathscr{U}(\phi_1(X))$ and let $\mathscr{U}(\phi_1(X),\phi_2(Y))$ denote the general inaccuracy measure of $Y$ with $X$, such that $\mathscr{U}(\phi_1,\phi_2)=\mathscr{U}(\phi_1(X)=\mathscr{U}(\phi_2(Y)$ for identical $\phi_i (\cdot)$. Another measure of interest is a class of divergence measures, denoted by $\mathscr{D}(\phi_1(X)| \phi_2(Y))$, given by 
\begin{equation} \label{eq:R5divergence}
\mathscr{D}(\phi_1(X)|\phi_2(Y))=\mathscr{U}(\phi_1 (X),\phi_2 (Y))-\mathscr{U}(\phi_1 (X)).
\end{equation}    
This general framework is evident in well-known divergences such as Kullback-Leibler (KL) divergence \citep{kullback1951information}, extropy divergence \citep{mohammadi2024dynamic}, cumulative extropy divergence \citep{saranya2025inaccuracy}, etc.  For example, when $\phi_1 (x) = f(x)$ and $\phi_2 (x) = g(x)$, the probability measures of $X$ and $Y$ respectively represent the PDF's of the non-negative and absolutely continuous random variables $X$ and $Y$, we get the Kullback-Leibler divergence, which can be represented as the difference between the Kerridge's inaccuracy \citep{kerridge1961inaccuracy}, $K(f, g) = -\int_{0}^{\infty}{f(x) \log g(x) dx}$ and the Shannon differential entropy \citep{shannon1948mathematical}, $H(F) = -\int_{0}^{\infty}{f(x) \log f(x) dx}$ as
 \begin{equation*}
   KL(f, g) = \int_{0}^{\infty}{f(x) \log \frac{f(x)}{g(x)}dx} = K(f,g) - H(f). 
 \end{equation*} 

As a complementary dual of entropy, extropy has emerged as a novel measure of uncertainty quantification. According to \cite{lad2015extropy}, the extropy of the non-negative random variable $X$ with PDF $f(\cdot)$ is defined as,
\begin{equation}\label{1.1}
    J(X)=-\frac{1}{2}\int_0^{\infty} f^2(x) dx.
\end{equation}
Motivated with the cumulative residual entropy (CRE) due to \cite{rao2004cumulative}, recently \cite{jahanshahi2020cumulative} introduced an alternative measure of uncertainty of a non-negative random variable $X$, called the survival extropy, defined as $J_s(X)=-\frac{1}{2}\int_{0}^{\infty}{\bar{F}}^2(x)dx$ obtained by replacing $f(x)$ in \eqref{1.1} by the survival function $\bar{F}(x) = P(X > x)$.  Based on the same idea, by using the distribution function $F(x) = P (X \leq x)$, \cite{nair2020dynamic} derived cumulative extropy as $\bar{\xi}J_s(X)=-\frac{1}{2}\int_{0}^{\infty}{F}^2(x)dx$.  \cite{toomaj2023extropy} derived some properties and several theoretical merits of extropy, and dynamic versions of extropy. 

Also, \cite{hashempour2024new} provided a new measure of inaccuracy based
on extropy for record statistics between distributions of the $n$th upper (lower) record value and the parent random variable, which is given in the form,
\begin{equation}\label{eq:R5extropyinaccuracy}
    \xi J(X,Y)=-\frac{1}{2}\int_0^\infty f(x)g(x)dx.
\end{equation} 
\cite{hashempour2024dynamic} defined cumulative past extropy inaccuracy (cumulative extropy inaccuracy) as $\bar{\xi}J(X,Y)=-\frac{1}{2}\int_0^\infty {F}(x)G(x)dx$.  Further, \cite{saranya2025inaccuracy} introduced the survival extropy inaccuracy of two non-negative random variables $X$ and $Y$ with SFs $\bar{F}(\cdot)$ and $\bar{G}(\cdot)$ respectively, as
 \begin{equation}
\xi J_s(X,Y)=-\frac{1}{2}\int_0^\infty \bar{F}(x)\bar{G}(x)dx.
 \end{equation}
Note that the cumulative extropy inaccuracy $\bar{\xi}J(X,Y)$ and survival extropy inaccuracy $\xi J_s(X,Y)$ are members of the class of inaccuracy measures $\mathscr{U}(\phi_1(X),\phi_2(Y))$ based on extropy, when $\phi_1(x) = F(x), \phi_2 (x) = G(x)$ and $\phi_1(x) = \bar{F}(x), \phi_2 (x) = \bar{G}(x)$ respectively.

A useful measure of similarity or dissimilarity between two random variables is the notion of relative extropy, proposed by \cite{lad2015extropy},
\begin{equation}\label{eq:R5relative extropy}
   d(f,g)=\frac{1}{2}\int_0^\infty \left({f(x)}-{g(x)}\right)^2 dx. 
\end{equation} 
Another member of the class of divergence measures \eqref{eq:R5divergence} based on extropy, to measure the discrimination between two non-negative random variables $X$ and $Y$ with PDFs $f(\cdot)$ and $g(\cdot)$, is given by
\begin{equation}\label{eq:R5discrimeasuredensity}
    J(f|g)=\frac{1}{2}\int_0^\infty [f(x)-g(x)]f(x)dx = \frac{1}{2}E_f[f(X)-g(X)] =\xi J(X,Y)-J(X).
\end{equation}
\cite{saranya2025inaccuracy} extended \eqref{eq:R5discrimeasuredensity} based on the SF, called the survival extropy divergence, defined by
    \begin{equation}\label{survival extropy}
SJ(\bar{F}|\bar{G})=\frac{1}{2}\int_0^\infty (\bar{F}(x)-\bar{G}(x))\bar{F}(x)dx.
         \end{equation}   
Clearly, \eqref{survival extropy} is also a member of the class of divergence measures \eqref{eq:R5divergence},
\begin{equation}
SJ(\bar{F}|\bar{G})=\xi J_s(X,Y)-J_s(X).    
\end{equation}

Note that the divergence class $\mathscr{D}(\phi_1(X)|\phi_2(Y))$ serves as a measure of how different $\phi_1(X)$ is from $\phi_2(Y)$, a way to quantify the "distance" between two probability measures. The divergence class between $Y$ and $X$ can also be defined in a similar manner, 
\begin{equation}\label{eq:R5reversedivergence}
    \mathscr{D}(\phi_2(X)|\phi_1(Y)) = \mathscr{U}(\phi_1(X),\phi_2(Y))- \mathscr{U}(\phi_2 (Y)).
\end{equation} 
Since class of divergence measures in \eqref{eq:R5divergence} and \eqref{eq:R5reversedivergence} are inherently asymmetric, symmetric divergence measures are constructed as the sum of divergences in both orders, leading to a class of symmetric divergences, say $\mathscr{SD}(\phi_1(X),\phi_2(Y))$ given by, 
\begin{equation}
   \mathscr{SD}(\phi_1(X),\phi_2(Y))= \mathscr{D}(\phi_1(X)|\phi_2(Y)) +  \mathscr{D}(\phi_2(Y)|\phi_1(X)).
\end{equation}
The class of symmetric divergence measures $\mathscr{SD}(\phi_1(X),\phi_2(Y))$ can also be expressed in terms of the class of inaccuracy measures $\mathscr{U}(\phi_1(X),\phi_2(Y))$ as
\begin{equation}\label{eq:R5symmetricdiv}
    \mathscr{SD}(\phi_1(X),\phi_2(Y))= 2\mathscr{U}(\phi_1(X),\phi_2(Y))-\mathscr{U}(\phi_1(X))-\mathscr{U}(\phi_2(Y)).
\end{equation}
The symmetric measures, such as relative extropy \eqref{eq:R5relative extropy} and relative cumulative extropy  \citep{saranya2024relative} are members of the above class \eqref{eq:R5symmetricdiv}.  Motivated by these symmetric measures, here we introduce two new classes of discrimination measures.
\newline
\newline
The article is organized as follows: Section 2 introduces the generalized divergence and similarity ratios based on extropies and their axioms. Examples for those measures satisfying the axioms are also discussed. In section 3, we derive the relationship between similarity measures with the cosines. Properties and relationships of defined measures with existing survival models, including the proportional hazard model, proportional reversed hazard model, etc, are discussed in Section 4. Estimation and simulation studies are included in Section 5. The application of similarity ratios in lifetime and image data is done in Section 6. 

\section{Extropy-based Generalized Divergence and Similarity Ratios}

Unlike the relative and divergence measures that quantify the difference between two probability models, this section proposes an extropy-based new generalized divergence and similarity ratios that quantify the discrepancy and closeness between two probability models. The following is the definition of the extropy-based generalized divergence ratio between two distributions.
\begin{definition}
Let $\phi_1(x)$ and $\phi_2(x)$ denote the general probability measures of $X$ and $Y$ respectively, such that $\phi_i, \; i = 1, 2$ can be a PDF, CDF, or SF. Further, let $U(\phi_1(X))$ be the generalized extropy of $X$, given by
\[
U(\phi_1(X))=-\frac{1}{2}\int_0^\infty \phi_1^2(x)dx,
\]
and
\[
U(\phi_1(X),\phi_2(Y))=-\frac{1}{2}\int_0^\infty \phi_1(x)\phi_2(x)dx
\]
denote the Generalized Extropy Inaccuracy (GEI) measure of $X$ and $Y$. Then the Generalized Extropy Divergence Ratio (GEDR) of $X$ with $Y$, $ I(\phi_1(X)|\phi_2(Y)) $ is given by 

\begin{equation}\label{eq:R5divergenceratio}
    I(\phi_1(X)|\phi_2(Y)) = \frac{U(\phi_1(X),\phi_2(Y))}{U(\phi_1(X))}.
\end{equation}   
Similarly, the generalized extropy divergence ratio of $Y$ with $X$ is given by  
\begin{equation}\label{2.2}
   I(\phi_2(Y)|\phi_1(X)) = \frac{U(\phi_1(X),\phi_2(Y))}{U(\phi_2(Y))}.
\end{equation} 
\end{definition}
Clearly, the GEI of $X$ and $Y$,  $U(\phi_1(X), \phi_2(Y))$ is symmetric and $U(\phi_1(X), \phi_1(X))=U(\phi_1(X))$, the generalized extropy of $X$.  Further, the GEDR is a dimensionless measure of divergence between the distributions of $X$ and $Y$ based on the probability measures $\phi_1$ and $\phi_2$, respectively.  The following axioms/properties are necessary and sufficient for a measure to be termed as GEDR.
\begin{enumerate}
 \item \textbf{Asymmetry}: GEDR is asymmetric with respect to $X$ and $Y$, since $ I(\phi_2(Y)|\phi_1(X))\not= I(\phi_1(X)|\phi_2(Y))$. 
    \item \textbf{Positivity}: GEDR is always positive since the ratio of two negative or two positive measures is always positive.
    \item \textbf{Unity under similarity}: $I(\phi_1(X)|\phi_2(Y)) = 1 = I(\phi_2(Y)|\phi_1(X))$ if the distributions of $X$ and $Y$ are identical.  
    \item \textbf{Boundedness}: $ I(\phi_2(Y)|\phi_1(X)) \in (0,\infty)$ if $U(\phi_1(X),\phi_2(Y))$ and $U(\phi_2(Y))$ are finite. Similarly the GEDR $I(\phi_1(X)|\phi_2(Y))\in (0,\infty)$ for finite $U(\phi_1(X),\phi_2(Y))$ and $U(\phi_1(X))$.
    \item Relationship between the order of the GEDR.
    \begin{equation}
         I(\phi_1(X)|\phi_2(Y)) > (<)1 \iff  I(\phi_2(Y)|\phi_2(X))< (>) 1 .
    \end{equation}
    \begin{proof}
        Applying axioms \ref{Axiom:R5similaritybound} and \ref{Axiom:R5similarity} of generalized extropy similarity ratio, we obtain 
        \[
        0<I(\phi_1(X)|\phi_2(Y))\times I(\phi_2(Y)|\phi_2(X))\leq 1.
        \]
        Then $I(\phi_1(X)|\phi_2(Y))>1$ implies $I(\phi_2(Y)|\phi_1(X))<1$.  The converse is similar.
    \end{proof}
    \item Upperbound of GEDR:  $I(\phi_1(X)|\phi_2(Y))\leq \sqrt{{U(\phi_2(Y))}/{U(\phi_1(X))}}$ and $I(\phi_2(Y)|\phi_1(X))\leq \sqrt{{U(\phi_1(X))}/{U(\phi_2(Y))}}$.
   \begin{proof}
      Using Cauchy-Schwarz inequality, we have 
      \[
       \int_0^\infty \phi_1(x)\phi_2(x)dx\leq \sqrt{\int_0^\infty \phi_1^2(x)dx}\sqrt{\int_0^\infty \phi_2^2(x)dx}.
      \]
      Assuming $\int_0^\infty \phi_1^2(x)dx>0$,
      \[
      \frac{\int_0^\infty \phi_1(x)\phi_2(x)dx}{\int_0^\infty \phi_1^2(x)dx}\leq \frac{\sqrt{\int_0^\infty \phi_1^2(x)dx}\sqrt{\int_0^\infty \phi_2^2(x)dx}}{\int_0^\infty \phi_1^2(x)dx},
      \]
      which completes the proof of the first bound.  The second case is similar.
    \end{proof}
   \item $U(\phi_1(X))\geq (\leq) U(\phi_2(Y)) \iff I(\phi_1(X)|\phi_2(Y))\geq (\leq)I(\phi_2(Y)|\phi_1(X)).$
   \begin{proof}
      The property can be directly obtained from the relationship:
      \[
      \frac{I(\phi_1(X)|\phi_2(Y))}{I(\phi_2(Y)|\phi_1(X))}=\frac{U(\phi_2(Y))}{U(\phi_1(X))}.
      \]
   Note that $U(\phi_1(X))$ and $U(\phi_2(Y))$ are always negative. 
   \end{proof}
   \end{enumerate}

As an alternative to symmetric divergence measures, we next define an extropy-based generalized similarity ratio, as follows.
\begin{definition}
Let $\phi_1(x)$ and $\phi_2(x)$ be the general probability measures,  $U(\phi_1(X))$ and $U(\phi_2(Y))$ be the generalized extropy of $X$ and $Y$ respectively. Further, let $U(\phi_1(X),\phi_2(Y))$ denote the GEI of $X$ and $Y$. Then the Generalized Extropy Similarity Ratio (GESR) of $X$ and $Y$ is given by 
\begin{equation} \label{eq:R5similarity_ratio}
\begin{split}
 S(\phi_1(X),\phi_2(Y))&=\frac{U(\phi_1 (X),\phi_2 (Y))^2}{U(\phi_1(X)) U(\phi_2(Y))}.
\end{split}
\end{equation}  
\end{definition}
$S(\phi_1(X),\phi_2(Y))$ provides a normalized measure of similarity between two distributions. This framework unifies different extropy-based divergence measures and offers a systematic way to analyze probabilistic discrepancies in various applications.

The following axioms/properties are necessary and sufficient for a measure to be termed the generalized extropy similarity ratio (GESR).
\begin{enumerate}
    \item \textbf{Symmetricity}: GESR is symmetric about $X$ and $Y$.
    \[
    S(\phi_1(X),\phi_2(Y))=S(\phi_1(Y),\phi_2(X)).
    \]
    \begin{proof}
         GEI such as extropy-based inaccuracy, survival extropy inaccuracy, and cumulative extropy inaccuracy are symmetric about $X$ and $Y$. That is $U(\phi_1(X),\phi_2(Y))=U(\phi_2(Y),\phi_1(X))$. So $S(\phi_1(X),\phi_2(Y))$ is symmetric about $X$ and $Y$.
    \end{proof}
     \item \textbf{Bounded in (0,1]} \label{Axiom:R5similaritybound}: $ S(\phi_1(X),\phi_2(Y))$ lies in the interval $(0,1]$.
     \begin{proof}
      
 The GESR always lies in the interval $(0,1]$. They are always positive as they are the ratio of two negatively valued functions.
Using Cauchy-Schwarz inequality, we have
\begin{equation}\label{R5:eqcauchyschwarz}
    \begin{split}
      \int_0^\infty \phi_1(x)\phi_2(x)dx \leq \int_0^\infty \phi^2_1(x)dx\int_0^\infty \phi^2_2(x)dx.
    \end{split}
\end{equation}
From \eqref{R5:eqcauchyschwarz},
\[
S(\phi_1(X),\phi_2(Y))=\frac{\int_0^\infty \phi_1(x)\phi_2(x)dx}{\int_0^\infty \phi^2_1(x)dx\int_0^\infty \phi^2_2(x)dx}\leq 1.
\]       
     \end{proof}
    \item  \textbf{Identity property}: $S(\phi_1(X),\phi_1(X))=1$.
    \begin{proof}
      Since $U(\phi_1(X),\phi_1(X))=U(\phi_1(X))$,  $S(\phi_1(X),\phi_1(X))=\frac{U(\phi_1(X))^2}{U(\phi_1(X))^2}=1$.
    \end{proof}
    \item \textbf{Unity under maximum similarity}: $ S(\phi_2(Y),\phi_1(X))=1$ (maximum similarity) if the distributions of $X$ and $Y$ are identical.
    \item \textbf{Similarity ratio is the product of corresponding divergence ratios.}\label{Axiom:R5similarity}
    \[
    S(\phi_1(X),\phi_2(Y))= I(\phi_1(X)|\phi_2(Y))* I(\phi_2(X)|\phi_1(Y))
    \]
\end{enumerate}
It is interesting to note that the similarity between the distributions is equal to the product of divergence or discrepancy between the distributions. While the symmetric divergence measure can be defined as the sum of the asymmetric divergence of $X$ and $Y$, and $Y$ and $X$, the GESR is obtained as the product of asymmetric GEDRs.  

The statistical terms covariance and correlation are two common measures to measure the linear relationship between two random variables. The covariance provides the direction of a linear association and measures the combined variability of the variables. Since it is scale-invariant, its magnitude depends on the units of the variables involved, which makes it difficult to interpret or compare across different contexts. Alternatively, we have the normalized version of covariance, called the correlation coefficient, which provides a dimensionless, bounded, and interpretable measure of the strength and direction of linear dependence. This normalization makes correlation a tool of choice when comparisons over different variables or datasets are needed.

A similar analogy applies in the context of existing divergence measures and defined similarity ratios. Like covariance, the divergence measures capture discrepancy or disagreement between two distributions but lack normalization, making it difficult to interpret directly across different distributions or functions. To address this limitation, the GESR is introduced. It serves a similar role to the correlation coefficient by providing normalized, dimensionless, and bounded measures of similarity between two informational structures. On the other hand, the GEDR is a positive, normalized form of divergence measure.  The GEDR measures the relative inaccuracy of $X$ under $Y$'s distribution compared to its own, while the GESR quantifies how closely two informational representations align. These measures elevate interpretability, comparability, and utility in both theoretical and applied settings, analogous to how the correlation coefficient improves upon covariance.

Specific forms of GEDR and GESR defined in \eqref{eq:R5divergenceratio} and \eqref{eq:R5similarity_ratio} based on different probability measures that satisfy the axioms are discussed in the sequel.
\begin{itemize}
    \item[(i)] \textbf{Extropy divergence and similarity ratio}
\begin{definition}
    For two absolutely continuous non-negative random variables $X$ and $Y$, the extropy divergence ratio is given by 
     \begin{equation}
        I_{E}(X|Y)=\frac{\xi J(X,Y)}{J(X)},
    \end{equation}
    and the extropy similarity ratio is given by
    \begin{equation}
       S_E(X,Y)=\frac{(\xi J(X,Y))^2}{J(X)J(Y)}=I_{E}(X|Y)\times I_{E}(Y|X).
    \end{equation}
\end{definition}
$I_{E}(X|Y)$ and $S_E(X,Y)$ are obtained by taking $\phi_1 (x) = f(x)$ and $\phi_2 (x) = g(x)$ in the expressions of GEDR in \eqref{eq:R5divergenceratio} and GESR in \eqref{eq:R5similarity_ratio} respectively.  $I_{E}(X|Y)$ can be obtained in a similar manner from \eqref{2.2} as $I_{E}(Y|X) = \frac{\xi J(X,Y)}{J(Y)}$.  
\begin{theorem}
    $S_E(X,Y)=1$ if and only if $f(x)=g(x)$ for all $x$.
\end{theorem}
    \begin{proof}
        \begin{equation}
            S_E(X,Y)=\frac{\left(\int_0^\infty f(x)g(x)dx\right)^2}{\int_0^\infty f^2(x)dx \int_0^\infty g^2(x)dx}
        \end{equation}
        Applying Cauchy-Schwarz inequality, we have
        \begin{equation}
        \begin{split}
          \left(\int_0^\infty f(x)g(x)dx\right)^2\leq \int_0^\infty f^2(x)dx \int_0^\infty g^2(x)dx, 
        \end{split}
        \end{equation}
        implies $S_E(X,Y)\leq 1$. Equality is achieved if and only if $f(x)$ and $g(x)$ are proportional, $i.e.$, $ f(x)=cg(x) $ for some constant $c>0$. 
        \[
        \int_0^\infty f(x)dx=c\int_0^\infty g(x)dx \implies c=1\implies f(x)=g(x) \forall x.
        \]
    \end{proof} 

   For two continous random variables $X$ and $Y$, we have  
  \[
  I_{E}(X|Y)>(<)1 \iff I_{E}(Y|X)<(>)1.
  \]
$I_{E}(X|Y)$ is a non-negative measure of divergence between the distributions of $X$ and $Y$ and is equal to 1 if $f(x)=g(x)$ for every $x$. The ratio of GEDR of both orders is equal to the ratio of corresponding extropies as follows:
 \[
 \frac{I_{E}(Y|X)}{{I_{E}(X|Y)}}=\frac{J(X)}{J(Y)}.
 \]
\item[(ii)] \textbf{Survival extropy divergence and similarity ratio} 
\begin{definition}
    For two absolutely continuous non-negative random variables $X$ and $Y$, the survival extropy divergence ratio of $X$ with $Y$ and $Y$ with $X$ are given by 
\[
I_{SE}(X|Y)=\frac{\xi J_s(X,Y)}{J_s(X)} \; \; \text{and} \; \; I_{SE}(Y|X) = \frac{\xi J_s(X,Y)}{J_s(Y)},
\]
and the survival extropy similarity ratio of $X$ and $Y$ is given by 
\[
S_{SE}(X,Y)=\frac{(\xi J_s(X,Y))^2}{J_s(X)J_s(Y)}=I_{SE}(X | Y) \times I_{SE}(Y | X)..
\] 
\end{definition}
The survival extropy similarity ratio is a positive measure of similarity between two distributions and is always symmetric about $X$ and $Y$. 
  
\item[(iii)] \textbf{{Cumulative extropy divergence and similarity ratio}}
\begin{definition}
  For two continuous random variables $X$ and $Y$, 
then the cumulative extropy divergence ratios are given by 
\[
I_{CE}(X|Y)=\frac{\bar{\xi}J(X,Y)}{\bar{\xi}J(X)} \; \; \text{and} \; \; I_{CE}(Y|X) = \frac{\bar{\xi}J(X,Y)}{\bar{\xi}J(Y)},
\]
and the cumulative extropy similarity ratio is given by 
\[
S_{CE}(X,Y)=\frac{(\bar{\xi}J(X,Y))^2}{\bar{\xi}J(X)\bar{\xi}J(Y)}=I_{CE}(X | Y) \times I_{CE}(Y | X)..
\]  
\end{definition}
The cumulative extropy similarity ratio is a positive measure of similarity between two distributions and is always symmetric about $X$ and $Y$. 
\end{itemize}

\section{Similarity ratio as square of cosine of angle between functions}

 Cosine similarity is one of the most popular similarity measures applied to text documents, such as in numerous information retrieval applications  \citep{baeza1999modern}) and clustering \citep{larsen1999fast}. Cosine similarity calculation between two vectors $P$ and $Q$ is performed by the following formula (see \cite{gomaa2013survey}).
\[
\cos \alpha = \frac{P\times Q}{|P|\times|Q|}=\frac{\sum_{i=1}^{n} p_i q_i}{\left( \sum_{i=1}^{n} p_i^2 \right)^{1/2} \left( \sum_{i=1}^{n} q_i^2 \right)^{1/2}}
\]

Now we define the GESR as the square of the cosine.
\begin{theorem}
 GESR is the square of the cosine of the angle between the probability measures $\phi_i$ of $X$ and $Y$.
\begin{equation}\label{eq:R5cossquaretheta}
 S(\phi_1(X),\phi_2(Y))=cos^2\theta,   
 \end{equation}
where $\theta$ is the angle between $\phi_1(X)$ and $\phi_2(Y)$.
\end{theorem}
\begin{proof}
In functional spaces, particularly the Hilbert space \( L^2([0, \infty)) \), the space of square-integrable functions, functions are treated analogously to vectors in Euclidean space. This means that notions like inner product, norm, and angle can be defined for functions.

Given two functions \( f(x) \) and \( g(x) \), we define:
\begin{itemize}
    \item[(i)] {Inner product:} $\langle f, g \rangle = \int_0^\infty f(x)g(x) \, dx$
     \item[(ii)] {Norms:} $\|f\| = \left( \int_0^\infty f^2(x) \, dx \right)^{1/2}, \quad \|g\| = \left( \int_0^\infty g^2(x) \, dx \right)^{1/2}$
     \item[(iii)] {Cosine of the angle} \( \theta \) between \( f \) and \( g \): $\cos\theta = \frac{\langle f, g \rangle}{\|f\| \|g\|}$
\end{itemize}

Thus, the cosine of the angle measures how "aligned" the two functions are in the function space. Replacing $f$ and $g$ by probability measures $\phi_1$ and $\phi_2$ and squaring, we get \eqref{eq:R5cossquaretheta} as
\[
S(\phi_1(X),\phi_2(Y))=\frac{U(\phi_1(X),\phi_2(Y))^2}{U(\phi_1^2(X))U(\phi_2^2(Y))}=\frac{\left(\int_0^\infty \phi_1(x)\phi_2(x)dx\right)^2}{\int_0^\infty \phi_1^2(x)\int_0^\infty\phi_2^2(x)dx}=cos^2\theta
\]
\end{proof} 
The square-of-cosine structure is consistent across:
\begin{itemize}
    \item[(i)] \textbf{Extropy similarity} \( S_E(X,Y) \) is based on probability density functions.
    The extropy-based similarity measures are formulated as:
\[
S_E(X,Y) = \left( \frac{ \langle f, g \rangle }{ \|f\| \|g\| } \right)^2= \cos^2\theta,
\]
\[
\cos \theta=\sqrt{S_E(X,Y)}
\]
    \item[(ii)] \textbf{Survival extropy similarity} \( S_{SE}(X,Y) \) is based on survival functions.
    The survival extropy similarity measures are formulated as:
\[
S_{SE}(X,Y) = \left( \frac{ \langle \bar{F}, \bar{G} \rangle }{ \|\bar{F}\| \|\bar{G}\| } \right)^2= \cos^2\theta,
\]
\[
\cos \theta=\sqrt{S_{SE}(X,Y)}
\]   
    \item [(iii)] \textbf{Cumulative extropy similarity} \( S_{CE}(X,Y) \) is based on cumulative distribution functions.
    The extropy-based similarity measures are formulated as:
\[
S_{CE}(X,Y) = \left( \frac{ \langle F, G \rangle }{ \|F\| \|G\| } \right)^2= \cos^2\theta,
\]
\[
\cos \theta=\sqrt{S_{CE}(X,Y)}
\]   
\end{itemize}  
Considering \( \phi_1 \) and \( \phi_2 \) as vectors in an infinite-dimensional space:
\begin{itemize}
    \item[(i)] \( \cos\theta = 1 \) ($i.e.$, \( \theta = 0^\circ \)) means the functions are perfectly aligned (identical up to scaling).
    \item[(ii)] \( \cos\theta = 0 \) ($i.e.$, \( \theta = 90^\circ \)) means the functions are orthogonal (completely dissimilar).
    \item[(iii)] Intermediate values of \( \cos\theta \) indicate partial similarity.
\end{itemize}
Consequently, \( S(X,Y) = (\cos\theta)^2 \) measures the degree of similarity between the shapes of two functions.
\begin{itemize}
    \item[(i)] High similarity ($i.e.$, \( S \) close to 1) means that two distributions have very similar behavior.
    \item [(ii)] Low similarity ($i.e.$, \( S \) close to 0) indicates that the two distributions are distinct in their structure.
\end{itemize}
The similarity measures are sensitive to the overall shape and spread of the distributions, not just pointwise differences. It enhances the difference more sharply: when two functions are slightly misaligned (small angle), their cosine drops gently, but squaring amplifies this drop.
Ultimately, this representation refines our understanding of how close or far two probability distributions are from each other, not only numerically but also geometrically.

\begin{example}[${S_E}$ and $S_{SE}$ of two exponential distributions]
Let $X \sim \text{Exp}(\lambda_1)$ and $Y \sim \text{Exp}(\lambda_2)$, with probability density functions $f(x) = \lambda_1 e^{-\lambda_1 x}, \quad g(x) = \lambda_2 e^{-\lambda_2 x}, \quad x \geq 0.$

The inner product of $f$ and $g$ is given by
\[
\langle f, g \rangle = \int_0^\infty f(x)g(x) \, dx = \lambda_1 \lambda_2 \int_0^\infty e^{-(\lambda_1+\lambda_2)x} \, dx.
\]
Evaluating the integral,
\[
\langle f, g \rangle = \lambda_1 \lambda_2 \times \frac{1}{\lambda_1+\lambda_2}.
\]
The norms of $f$ and $g$ are
\[
\|f\|^2 = \int_0^\infty f^2(x) \, dx = \lambda_1^2 \int_0^\infty e^{-2\lambda_1 x} \, dx = \lambda_1^2 \times \frac{1}{2\lambda_1} = \frac{\lambda_1}{2},
\]
\[
\|g\|^2=\frac{\lambda_2}{2}.
\]
Hence, the $S_E$  between $f$ and $g$ is
\[
\cos^2\theta = \left(\frac{\langle f, g \rangle}{\|f\| \|g\|}\right)^2 = \frac{\left(\dfrac{\lambda_1 \lambda_2}{\lambda_1+\lambda_2}\right)^2}{{\dfrac{\lambda_1}{2}} {\dfrac{\lambda_2}{2}}} = \frac{4 {\lambda_1 \lambda_2}}{(\lambda_1+\lambda_2)^2}.
\]
Similarly, the survival extropy similarity $S_{SE}$ can be derived and is equal to ${4 {\lambda_1 \lambda_2}}/{(\lambda_1+\lambda_2)^2}$.
Therefore, the similarity is explicitly on the rates $\lambda_1$ and $\lambda_2$.
 
When $\lambda_1 = \lambda_2$, we have $\cos\theta = 1$ and $\theta = 0^\circ$, meaning the two distributions are identical. When $\lambda_1$ and $\lambda_2$ are very different, $\cos\theta$ decreases, and $\theta$ increases, indicating greater dissimilarity between the two distributions.
\begin{figure}
    \centering
    \includegraphics[width=0.7\linewidth]{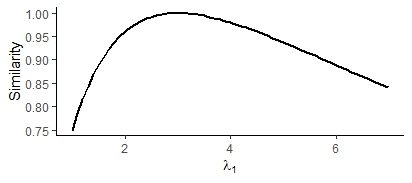}
    \caption{$S_E$ and $S_{SE}$ of two exponential distributions with parameters $\lambda_1$ and $\lambda_2$ over different $\lambda_1$ values ($\lambda_2=3$).}
    \label{fig:R5similarityofexp}
\end{figure}

Figure \ref{fig:R5similarityofexp} shows the $S_E$ and $S_{SE}$ of two exponential distributions with parameters $\lambda_1$ and $\lambda_2=3$ respectively over different values of $\lambda_1$. Similarity ratio is maximum when $\lambda_1=\lambda_2$ and decreases when $\lambda_2$ varies from 3.
\end{example}

\section{More properties}

The following theorem derives the scale invariance of the GEDR and GESR.
\begin{theorem}\label{Theorem:R5scaleinvariance} (\textbf{Scale invariance property}):
The GEDR and GESR are scale invariant. Let $X$ and $Y$ be two continuous random variables and let $a>0$ be a constant. Define $X_a = aX$ and $Y_a = aY$. Then, the GEDR and GESR satisfy :
\[
I(aX|aY)=I(X|Y), \qquad S(aX,aY)=S(X,Y).
\]
\end{theorem}

\begin{proof}
Let $f_a(x)$ denote the density of $X_a = aX$. Then $f_a(x) = \frac{1}{a}f\left(\frac{x}{a}\right), \quad g_a(x) = \frac{1}{a}g\left(\frac{x}{a}\right).$
The extropy of $X_a$ is: $ J(X_a) = -\frac{1}{2} \int_0^\infty f_a^2(x) \, dx 
= -\frac{1}{2} \int_0^\infty \left( \frac{1}{a} f\left( \frac{x}{a} \right) \right)^2 dx.$
Substitute $u = \frac{x}{a}$, so that $dx = a\, du$, we get:$$J(X_a) = -\frac{1}{2} \cdot \frac{1}{a^2} \cdot a \int_0^\infty f^2(u)\, du = \frac{1}{a} J(X).$$
Similarly, $J(Y_a) = \frac{1}{a} J(Y)$. Now, compute the inaccuracy term:
\[
\xi J(X_a, Y_a) = -\frac{1}{2} \int_0^\infty f_a(x) g_a(x)\, dx 
= -\frac{1}{2} \int_0^\infty \frac{1}{a^2} f\left(\frac{x}{a}\right) g\left(\frac{x}{a}\right) dx.
\]
Substituting $u = \frac{x}{a}$ again,
\[
\xi J(X_a, Y_a) = -\frac{1}{2a} \int_0^\infty f(u)g(u)\, du = \frac{1}{a} \xi J(X,Y).
\]
Now the divergence ratio becomes:
\[
I_E(X_a|Y_a) = \frac{\xi J(X_a, Y_a)}{J(X_a)} = \frac{\frac{1}{a} \xi J(X,Y)}{\frac{1}{a} J(X)} = I_E(X|Y).
\]
Also, $I_E(Y_a|X_a) = I_E(Y|X).$ Therefore, the similarity ratio:
\[
S_E(X_a, Y_a) = I_E(X_a|Y_a) \cdot I_E(Y_a|X_a) = I_E(X|Y) \cdot I_E(Y|X) = S_E(X,Y).
\]
Similarly, the survival extropy of $X_a$ is
\[
J_s(X_a) = -\frac{1}{2} \int_0^\infty \bar{F}_a^2(x) \, dx = -\frac{1}{2} \int_0^\infty \bar{F}^2\left( \frac{x}{a} \right) dx.
\]
Substituting $u = \frac{x}{a}$, so that $dx = a \, du$, we obtain $J_s(X_a) = -\frac{1}{2} a \int_0^\infty \bar{F}^2(u) \, du = a J_s(X).$
Similarly, the survival extropy inaccuracy between $X_a$ and $Y_a$ is
\[
\xi J_s(X_a,Y_a) = -\frac{1}{2} \int_0^\infty \bar{F}_a(x)\bar{G}_a(x) \, dx = -\frac{1}{2} a \int_0^\infty \bar{F}(u)\bar{G}(u) \, du = a \xi J_s(X,Y).
\]
Therefore,
\[
I_{SE}(X_a|Y_a) = \frac{\xi J_s(X_a,Y_a)}{J_s(X_a)} = \frac{a \xi J_s(X,Y)}{a J_s(X)} = \frac{\xi J_s(X,Y)}{J_s(X)} = I_{SE}(X|Y).
\]
We obtain $S_{SE}(X_a,Y_a) = I_{SE}(X_a|Y_a) I_{SE}(Y_a|X_a) = I_{SE}(X|Y) I_{SE}(Y|X) = S_{SE}(X,Y).$
Similarly, we get the cumulative extropy divegence as
\[
I_{CE}(X_a|Y_a) =  I_{CE}(X|Y),
\]
which implies $S_{CE}(X_a,Y_a) = I_{CE}(X_a|Y_a) I_{CE}(Y_a|X_a) = I_{CE}(X|Y) I_{CE}(Y|X) = S_{CE}(X,Y).$

Thus, the survival and cumulative extropy divergence and similarity ratios are scale invariant.
This completes the proof.
\end{proof}
The following theorem examines the location invariance property of the GESR.
\begin{theorem}\label{Theorem:R5locationinvariance}(\textbf{Location invariance under common shift}): The GEDR and GESR are location invariant. Consider two absolutely continuous nonnegative random variables $X$ and $Y$.
Let \( X' = X+a \) and \( Y' = Y+a \) for some constant \( a \in \mathbb{R} \). Then, the GEDR and GESR satisfy
\[
I(X'|Y')=I(X|Y), \qquad S(X',Y') = S(X,Y).
\]
\end{theorem}

\begin{proof}
Shifting $X$ and $Y$ by \(a\), we get $f_{X'}(x) = f(x-a)$ and $g_{Y'}(x) = g(x-a).$
Therefore,
\[
\xi J(X',Y') = -\frac{1}{2}\int_0^\infty f(x-a)g(x-a)\,dx.
\]
Substituting \( u = x-a \) (so \( x = u+a \), \( dx=du \)), we obtain
\[
\xi J(X',Y') = -\frac{1}{2}\int_{-a}^\infty f(u)g(u)\,du.
\]
Using the assumption \( f(u) = g(u) = 0 \) for \( u < 0 \), this reduces to
\[
\xi J(X',Y') = -\frac{1}{2}\int_0^\infty f(u)g(u)\,du = \xi J(X,Y).
\]
Similarly, we find $J(X') = -\frac{1}{2}\int_0^\infty f^2(x-a)\,dx = -\frac{1}{2}\int_{-a}^\infty f^2(u)\,du = J(X),$
and $J(Y') = J(Y).$ Therefore,
\[
S_E(X',Y') = \frac{(\xi J(X',Y'))^2}{J(X')J(Y')} = \frac{(\xi J(X,Y))^2}{J(X)J(Y)} = S_E(X,Y).
\]

The same reasoning applies for the cumulative extropy similarity \( S_{CE}(X,Y) \) and survival extropy similarity \( S_{SE}(X,Y) \), because they also consists of integrals of products of shifted cumulative distribution functions or survival functions, respectively, and the shift affect them similarly. So, we obtain
\[
S_{SE}(X',Y')=S_{SE}(X,Y); \quad S_{CE}(X',Y')=S_{CE}(X,Y).
\]
\end{proof}

\begin{corollary}
    Let the probability measure of $X$ is $\phi_1(x)$. Then the similarity ratio between $X$ and $aX$ have the form
    \[
    S(X,aX)=a\left(I(\phi_1(X)|\phi_1(aX))\right)^2.
    \]
\end{corollary}
\begin{proof}
The proof is directly obtained from the given equation:
    \[
     S(X,aX)=\frac{a\left(\int_0^\infty \phi_1(ax)\phi_1(x)dx\right)^2}{\left(\int_0^\infty \phi_1^2(x)dx\right)^2}=\frac{aU(\phi_1(aX),\phi_1(X))^2}{U(\phi_1(X))^2}.
    \]
\end{proof}
The following theorem gives bounds for the extropy divergence ratio and the survival extropy divergence ratio under the proportional hazards model (PHM).

\begin{theorem}\label{Theorem:R5_SE_prophazardmodel}
 Let \( X \) and \( Y \) be two absolutely continuous nonnegative random variables with PDFs \( f(x) \) and \( g(x) \), respectively, satisfying the PHM $h_Y(x) = c h_X(x),$ for some constant $c > 0$ (\cite{cox1972regression}) and $h_X(x) = \frac{f_(x)}{\bar{F} (x)}$ and $h_Y(x) = \frac{g(x)}{\bar{G}(x)}$ respectively denote the hazard rates of $X$ and $Y$.  For \( c >(<) 1\), we have the following relationships.
   \begin{itemize}
     \item [a).]  
        \begin{enumerate}
            \item[(i)] $I_{E}(X|Y) <(>) c^2 I_{E}(Y|X)$.
            \item [(ii)] $J(Y) >(<) c^2 J(X)$.
            \item[(iii)] $I_{E}(X|Y)<(>)c$ and  $I_{E}(Y|X)>(<)\frac{1}{c}$.
        \end{enumerate}  
    \item [b).]  
        \begin{enumerate}
            \item[(i)] $I_{SE}(X|Y) <(>) I_{SE}(Y|X)$.
            \item [(ii)] $J_s(X) <(>) J_s(Y)$.
            \item[(iii)] $I_{SE}(X|Y)<(>)1\ \& \ I_{SE}(Y|X)>(<)1$.
        \end{enumerate} 
        \end{itemize}
 \end{theorem}
\begin{proof}
a) Assume that $X$ and $Y$ satisfy the PHM.             \begin{equation*}
            h_Y(x)=ch_X(x) \; \text{equivalently} \; \bar{G}(x)=\bar{F}(x)^c.
        \end{equation*}
First, we consider extropy divergence ratio $I_E(X|Y)$. Since $g(x)=h_Y(x)\bar{G}(x)=ch_X(x)\bar{F}^c(x)$ and $f(x)=h_X(x)\bar{F}(x)=\frac{h_Y(x)}{c}\bar{G}^{\frac{1}{c}}(x)$, we have
                \begin{equation}\label{eq:R5ratioextropydivergenceratio}
               \frac{I_E(X|Y)}{I_E(Y|X)}=\frac{J(Y)}{J(X)} = \frac{\int_0^\infty {g^2(x)dx}}{\int_0^\infty {f^2(x)dx}} =\frac{c^2\int_0^\infty {h^2_Y(x)}\bar{G}^{2}(x)dx}{\int_0^\infty h^2_Y(x){\bar{G}^{\frac{2}{c}}(x)dx}}.
                \end{equation}
                and
                \begin{equation}\label{eq:R5extropydivergenceratio_pro}
                    I_E(X|Y)=\frac{\int_0^\infty f(x)g(x)dx}{\int_0^\infty f^2(x)dx}=\frac{c\int_0^\infty h^2_X(x)\bar{F}^{c+1}(x)dx}{\int_0^\infty h^2_X(x)\bar{F}^{2}(x)dx}.
                \end{equation}
                
                \begin{itemize}
                
                    \item[(i)] Let $c>1$.  Using \eqref{eq:R5ratioextropydivergenceratio},
\begin{equation*}
\frac{1}{c}<1\implies\frac{2}{c}<2 \implies \bar{G}(x)^{\frac{2}{c}} > \bar{G}^2(x) \implies \frac{\int_0^\infty {h^2_Y(x)}\bar{G}^{2}(x)dx}{\int_0^\infty h^2_Y(x){\bar{G}^{\frac{2}{c}}(x)dx}}<1.
\end{equation*}
     Applying in  \eqref{eq:R5ratioextropydivergenceratio}, 
     \[
     \frac{I_E(X|Y)}{I_E(Y|X)}<c^2.
     \]
     Similarly, for $c<1$, $ \frac{I_E(X|Y)}{I_E(Y|X)}>c^2.$
    \item[(ii)] From \eqref{eq:R5ratioextropydivergenceratio}, for $c>(<)1$, $J(Y)>(<)c^2J(X)$.

     \item[(iii)] For $c>1$, 
     \begin{equation*}
         c+1<2\implies \bar{F}^{c+1}(x)< \bar{F}^2(x) \implies \frac{\int_0^\infty h^2_X(x)\bar{F}^{c+1}(x)dx}{\int_0^\infty h^2_X(x)\bar{F}^{2}(x)dx}<1. .
     \end{equation*}
     
     Applying in \eqref{eq:R5extropydivergenceratio_pro}, we get $I_E(X|Y)<c$. Similarly, for $c<1$, it becomes $I_E(X|Y)>c$. The steps are the same to prove the case of $I_E(Y|X)$.
     \end{itemize} 

b) Next, we consider the survival extropy divergence ratio $I_{SE}$ under the PHM.  We have 
\begin{equation}\label{eq:R5ratiosurvexdivratio}
            \frac{I_{SE}(X|Y)}{I_{SE}(Y|X)}=\frac{\int_0^\infty (\bar{F}(x))^{2c}dx}{\int_0^\infty (\bar{F}(x))^{2}dx}
        \end{equation}
 and       \begin{equation}\label{eq:R5survexdiveratioprop}
            I_{SE}(X|Y)=\frac{\int_0^\infty \bar{F}(x)\bar{G}(x)dx}{\int_0^\infty \bar{F}^2(x)dx}=\frac{\int_0^\infty (\bar{F}(x))^{c+1}dx}{\int_0^\infty (\bar{F}(x))^{2}dx}.
        \end{equation}
        \begin{itemize}
            \item[(i)] Let $c>1$. $$c>1\implies2c>2\implies \bar{F}^{2c}(x)<\bar{F}^2(x).$$ Applying in  \eqref{eq:R5ratiosurvexdivratio}, we get $\frac{I_{SE}(X|Y)}{I_{SE}(Y|X)}<1$, which proves $I_{SE}(X|Y)<I_{SE}(Y|X)$. Now for $c<1$, the relationship becomes $I_{SE}(X|Y)>I_{SE}(Y|X)$.
            \item[(ii)] For $c>1$, we obtain $\frac{I_{SE}(X|Y)}{I_{SE}(Y|X)}=\frac{J_s(Y)}{J_s(X)}<1$, implies $J_s(Y) >J_s(X)$ since $J_s(.)<0$. And $J_s(Y) <J_s(X)$ for $c<1$.
            \item[(iii)] Using \eqref{eq:R5survexdiveratioprop} and for $c>1$, $c+1>2\implies \bar{F}^{c+1}(x)<\bar{F}^2(x) \implies I_{SE}(X|Y)<1$. Similarly for $c<1$, we derive $I_{SE}(X|Y)>1$. Also, $I_{SE}(Y|X)>1$ for $c>1$ and $I_{SE}(Y|X)<1$ for $c<1$.
        \end{itemize}
\end{proof}
From Theorem \ref{Theorem:R5_SE_prophazardmodel}, the corresponding similarity ratios attain their maximum value 1 if and only if $X \stackrel{d}{=} Y$. Moreover, both $ S_E(X, Y) $ and $S_{SE}(X, Y)$ decrease as \( |c - 1| \) increases. 

    The following theorem gives bounds for the CEDR and cumulative extropy ordering under the proportional reversed hazards model (PRHM).
    \begin{theorem}\label{Theorem:R5proprevhazard}
Let $X$ and $Y$ be two absolutely continuous nonnegative random variables that satisfy PRHM (\cite{gupta1998modeling}), $\lambda_Y(x)=c\lambda_X(x),\; c>0$.   Then for \( c >(<) 1 \),  the following relationships hold,
        \begin{enumerate}
            \item[(i)] $I_{CE}(X|Y) <(>) I_{CE}(Y|X)$.
            \item[(ii)] $\bar{\xi}J(X)<(>) \bar{\xi}J(Y)$.
            \item[(iii)] $I_{CE}(X|Y)<(>)1 \ \& \ I_{CE}(Y|X)>(<)1$.
        \end{enumerate}  
    \end{theorem}
    \begin{proof}
        We have $G(x)=(F(x))^c$ under the PRHM.
        \begin{equation}\label{eq:R5ratiocumdivratio}
            \frac{I_{CE}(X|Y)}{I_{CE}(Y|X)}=\frac{\int_0^\infty (G(x))^2dx}{\int_0^\infty (F(x))^2dx}=\frac{\int_0^\infty (G(x))^2dx}{\int_0^\infty (G(x))^{2c}dx}
        \end{equation}
        and
        \begin{equation}\label{eq:R5cumdivratioproprev}
            I_{CE}(X|Y)=\frac{\int_0^\infty F(x)G(x)dx}{\int_0^\infty (F(x))^2dx}=\frac{\int_0^\infty (G(x))^{c+1}dx}{\int_0^\infty (G(x))^{2c}dx}
        \end{equation}
        Since $0\leq G(x)\leq 1$, for $c>1$, we get \eqref{eq:R5ratiocumdivratio} becomes greater than 1, which implies $I_{CE}(X|Y)>I_{CE}(Y|X)$, which further implies $\hat{\xi}J(X)\leq \bar{\xi}J(Y)$ From \eqref{eq:R5cumdivratioproprev}, for $c>1$, we obtain $I_{CE}(X|Y)>1.$ Similarly, the inequality is reversed for $0<c<1.$ This completes the proof.
    \end{proof}
    From Theorem \ref{Theorem:R5proprevhazard}, under PRHM, the cumulative extropy similarity measure satisfies
\[
S_{CE}(X, Y) = \frac{\left[\int_0^\infty (G(x))^{c+1} dx\right]^2}{\left[\int_0^\infty (G(x))^{2c} dx\right] \cdot \left[\int_0^\infty (G(x))^2 dx\right]} \leq 1,
\]
with equality if and only if \( c = 1 \), $i.e.$, $X \stackrel{d}{=} Y$. Furthermore, $S_{CE}(X,Y)$ decreases as $|c - 1|$ increases.

    \begin{theorem} \label{theorem:R5charpropuncertmodel}
         Two absolutely continuous nonnegative random variables $X$ and $Y$ satisfy the proportional generalized extropy model, 
        \[
        U(\phi_2(Y))=c \times U(\phi_1(X)), \; c>0,
        \]
      if and only if $X$ and $Y$ satisfy the proportional GEDR model:
        \[
        I(\phi_2(Y)|\phi_1(X))=\frac{1}{c}I(\phi_1(X)|\phi_2(Y)).
        \]
    \end{theorem} 
    \begin{corollary}
      Under Theorem \ref{theorem:R5charpropuncertmodel}, 
      \begin{enumerate}
          \item [(i)] $S(\phi_1(X),\phi_2(Y))=\frac{1}{c}(I(\phi_1(X)|\phi_2(Y)))^2.$ 
          \item [(ii)] $I(\phi_1(X)|\phi_2(Y))<\sqrt{c}.$
          \item [(iii)] $I(\phi_2(Y)|\phi_1(X))<\frac{1}{\sqrt{c}}.$
      \end{enumerate}       
    \end{corollary}

In the following section, we derive some non-parametric estimators for all similarity ratios discussed so far and conduct some simulation studies to validate their performance.
\section{Estimation and Simulation Studies}

Let $X_1, X_2, \dots, X_n \sim f$ and $Y_1, Y_2, \dots, Y_n \sim g$ be independent samples from two continuous distributions with densities $f$ and $g$, respectively. To estimate $S_E(X, Y)$ from data, we use the kernel density estimators by \cite{parzen1962estimation}:
\[
\hat{f}_n(x)=\frac{1}{nb_n}\sum_{j=1}^{n}k\left(\frac{x-X_j}{b_n}\right), \quad
\hat{g}_n(x)=\frac{1}{nb_n}\sum_{j=1}^{n}k\left(\frac{x-Y_j}{b_n}\right),
\]
where $k(\cdot)$ is a kernel of order $s$ with compact support and ${b_n}$, the bandwidths, is a sequence of positive numbers such that
$b_n\rightarrow 0$ and $nb_n\rightarrow\infty$ as $n\rightarrow \infty$.

Then, the non-parametric kernel-based estimator of the extropy similarity $S_E$ is given by
\[
\hat{S}_E = \frac{\left( \int_0^\infty \hat{f}_h(x)\hat{g}_h(x)dx \right)^2}{\left( \int_0^\infty \hat{f}_h^2(x)dx \right) \left( \int_0^\infty \hat{g}_h^2(x)dx \right)}.
\]

To empirically estimate the survival and cumulative extropy-based similarity and divergence measures, we approximate the integral definitions by summations using the empirical distribution and survival functions.

Let\ $\hat{\bar{F}}_n(x) = \frac{1}{n} \sum_{i=1}^n \mathbf{1}\{X_i > x\}$ and $\hat{\bar{G}}_n(x) = \frac{1}{m} \sum_{j=1}^m \mathbf{1}\{Y_j > x\}$ be the empirical survival functions. We considered a pooled sample ($Z$) of values of $X$ and $Y$ to estimate the inaccuracy and the combined sample size $N=m+n$. Then,   
    \[
    \widehat{\xi J_s}(X,Y) = -\frac{1}{2} \sum_{k=1}^{N} \hat{\bar{F}}(z_k) \hat{\bar{G}}(z_k) \Delta z_k,
    \]   
    \[
    \widehat{J_s}(X) = -\frac{1}{2} \sum_{k=1}^{n} \hat{\bar{F}}^2(x_k) \Delta x_k,
    \] 
    \[
     \widehat{J_s}(Y) = -\frac{1}{2} \sum_{k=1}^{m} \hat{\bar{G}}^2(y_k) \Delta y_k.
    \]
   The empirical survival extropy similarity ratio is given by :
    \[
    \widehat{S}_{SE}(X,Y) = \frac{\left(\widehat{\xi J_s}(X,Y)\right)^2}{\widehat{J_s}(X)\widehat{J_s}(Y)}
    \]

Let $\hat{F}(x) = \frac{1}{n} \sum_{i=1}^n \mathbf{1}\{X_i \leq x\}$ and $\hat{G}(x) = \frac{1}{m} \sum_{j=1}^m \mathbf{1}\{Y_j \leq x\}$ be the empirical cumulative distribution functions. Then the cumulative extropy inaccuracy estimator is 
    \[
    \widehat{\bar{\xi}J}(X,Y) = -\frac{1}{2} \sum_{k=1}^{N} \hat{F}(z_k) \hat{G}(z_k) \Delta z_k.
    \]
We have:
    \[
    \widehat{\bar{\xi}J}(X) = -\frac{1}{2} \sum_{k=1}^{n} \hat{F}^2(x_k) \Delta x_k,
    \]
     \[
    \widehat{\bar{\xi}J}(Y) = -\frac{1}{2} \sum_{k=1}^{m} \hat{G}^2(y_k) \Delta y_k.
    \]
    
   Then the empirical cumulative extropy similarity ratio is given by
    \[
    \widehat{S}_{CE}(X,Y) = \frac{\left(\widehat{\bar{\xi}J}(X,Y)\right)^2}{\widehat{\bar{\xi}J}(X)\widehat{\bar{\xi}J}(Y)}.
    \]
    
Tables \ref{Table:R5extropysimesti}, \ref{Table:R5survextropysimesti}, and \ref{Table:R5cumextropysimesti} present the bias and mean squared error (MSE) of the estimated extropy similarity $\hat{S}_E(X,Y)$, survival extropy similarity $\hat{S}_{SE}(X,Y)$, and cumulative extropy similarity $\hat{S}_{CE}(X,Y)$, respectively, for increasing sample sizes. The estimates were computed for pairs of distributions with known theoretical similarity values: two beta distributions ($\beta(3,2)$ and $\beta(2,3)$) for $S_E(X,Y)$, exponential distributions with rates $\lambda_X=1$ and $\lambda_Y=2$ for $S_{SE}(X,Y)$, and a uniform distribution $U(0,1)$ with a beta distribution ($\beta(3,2)$) for $S_{CE}(X,Y)$. As observed, all three similarity estimators exhibit decreasing bias and MSE as the sample size increases, confirming their consistency.
\begin{table}[]
    \centering
    \caption{Bias and MSE of $S_E(X,Y)$ of two beta distributions ( $\beta(3,2)$ and $\beta(2,3)$ ) with actual value 0.5625}
    \begin{tabular}{p{2cm}|p{2cm}|p{2cm}|p{3cm}}
    \hline
       $n$  & $\hat{S}_E(X,Y)$ & Bias & MSE  \\
       \hline
        50 & 0.6073 & 0.0448 & 0.00201\\
        75 & 0.5396 & 0.0229 & 0.00053 \\
        100 & 0.5423 & 0.0202 & 0.00041\\
        200 & 0.5576 & 0.0049& $2.4\times 10^{-5}$\\
        
        \hline
    \end{tabular}
    
    \label{Table:R5extropysimesti}
\end{table}
\begin{table}[]
    \centering
    \caption{Bias and MSE of $S_{SE}(X,Y)$ of two exponential distributions ($\lambda_X=1$, $\lambda_Y=2$) with actual value 0.8889}
    \begin{tabular}{p{2cm}|p{2cm}|p{2cm}|p{3cm}}
    \hline
       $n$  & $\hat{S}_{SE}(X,Y)$ & Bias & MSE  \\
       \hline
        50 & 0.8764 & 0.0125 & 0.000156\\
        75 & 0.8778 & 0.0111 & 0.000123 \\
        100 & 0.8802 & 0.0087 & $7.56\times 10^{-5}$\\
        200 & 0.8887 & 0.0002 & $2\times 10^{-7}$\\
        
        \hline
    \end{tabular}
    
    \label{Table:R5survextropysimesti}
\end{table}
\begin{table}[h]
    \centering
     \caption{Bias and MSE of $S_{CE}(X,Y)$ of uniform $U(0,1)$ and beta ($\beta(3,2)$) distributions with actual value 0.94501}
    \begin{tabular}{p{2cm}|p{2cm}|p{2cm}|p{3cm}}
    \hline
       $n$  & $\hat{S}_{CE}(X,Y)$ & Bias & MSE  \\
       \hline
        50 & 0.9715 & 0.02649 & 0.0007\\
         75 & 0.9702 & 0.02518 & 0.0006\\
          100 & 0.9667 & 0.02167 & 0.0005 \\
        200 & 0.9628 & 0.0178 & 0.0004\\
        \hline
    \end{tabular}
   
    \label{Table:R5cumextropysimesti}
\end{table}


\begin{table}[h!]
\begin{center}
\caption{Scale and location invariance of similarity measures using estimates. (Note that $X$ and $Y$ follows exponential with $\lambda_X=1$ and $\lambda_Y=2$ respectively.) }
\label{Table:R5invariancesimilarity} 
\begin{tabular}{p{2cm}|p{2cm}|p{2cm}|p{2cm}|p{2cm}}
 \hline
 A&B & $\hat{S}_{SE}(A,B)$& $\hat{S}_{CE}(A,B)$&$\hat{S}_E(A,B)$\\ 
 \hline
$X$&$Y$& 0.8744364& 0.8857035&0.6657251\\
$2X$&$2Y$&0.8744364&0.8857035&0.6657251\\
$X/2$&$Y/2$& 0.8744364&0.8857035&0.6657251\\
$X+0.5$& $Y+0.5$& 0.8744364&0.8857035&0.6674945\\
$X-0.5$ & $Y-0.5$ &0.8744364&0.8857035&0.6476732\\
 \hline
\end{tabular}
\end{center}
\end{table}
\subsection{Location and scale invariance of estimators}
Table \ref{Table:R5invariancesimilarity} shows the estimated values of all three defined similarity ratios for two exponential random variables $X$ and $Y$ with parameters $\lambda_1$ and $\lambda_2$ respectively. The estimates are computed for the original pair $(X, Y)$ as well as for scaled ($2X$, $2Y$ and $X/2$, $Y/2$) and shifted ($X \pm 0.5$, $Y \pm 0.5$) versions of the random variables. The estimates of the first three rows of the Table ~\ref{Table:R5invariancesimilarity} show that all three similarity ratios remain unchanged under scale transformations, confirming the scale-invariance property. The subsequent rows shows that the survival and cumulative extropy similarities also remain constant under location shifts. The extropy similarity ratio $\hat{S}_E(A,B)$ exhibits minor variation under shifting, remaining close to its original value, which underscores its relative stability for location changes. These findings substantiate Theorems \ref{Theorem:R5locationinvariance} and \ref{Theorem:R5scaleinvariance} in which the extropy-based similarity ratios are invariant in scale or location changes. They can be used for comparing distributions that may differ in scale and location parameters.

 \section{Data analysis}
\subsection{Lifetime data analysis using similarity ratios}
We consider lifetime data to explore similarity ratios. The data on the time to death of mice receiving various doses of red dye No. 40 by \cite{lagakos1981case} consist of the time to death data of 400 mice subjected to a lifetime feeding experiment. The objective of the experiment was to investigate the oncogenicity (cancer-causing potential) of a food colour labelled as FD $\&$ C Red No. 40. Fifty mice from each gender (female and male) were allotted to each of four groups: a control group and three dose-level groups (low, medium and high) of the test substance. The variable of interest is the observed lifetime of mice exposed to various dose-level groups. Using similarity ratios, we study the similarity of the lifetime pattern of mice under different treatment groups. Table \ref{Table:R5micedatasimilarity} gives the estimated values of all three similarity measures for different combinations of dose levels. All similarity estimates of time to death between different treatment groups show a consistent trend. All three similarity measures give values ranging from 0 to 1, where higher values indicate greater similarity between the distributional characteristics of two groups. Each group shows perfect similarity with itself as expected. However, when comparing the Control group with the increasing dose groups (Low, Medium and High), a progressive decline in similarity values is observed across all estimators, indicating that the treatment dose affects the underlying distribution of the data. For example, consider $\hat{S}_{SE}(X,Y)$. It decreases from 0.96298 (Control vs Low) to 0.94205 (Control vs Medium) and further to 0.90603 (Control vs High). Similar trends are observed with the other two measures. The high-dose group consistently shows the lowest similarity to the Control group, revealing a significant departure from the distribution. Comparisons among adjacent treatment groups (e.g., Low vs Medium, Medium vs High) show moderate similarity values, confirming a gradual, dose-dependent shift in distribution. On the whole, the similarity estimates quantitatively confirm that increasing doses of red dye No. 40 result in consistent changes in the data distribution. Figure \ref{fig:R5extropysimilarityheatmap_micedata} is a heatmap of extropy similarity, visualizing the trend of similarity of distributions of time to death between different dose levels.

\begin{table}[h!]
\centering
\caption{Similarity estimates between different treatment groups of mice data}
\label{Table:R5micedatasimilarity} 
\scalebox{0.9}{
\begin{tabular}{c|c|c|c|c|c}
\hline
\textbf{} & \textbf{Similarity} & \textbf{Control} & \textbf{Low} & \textbf{Medium} & \textbf{High} \\
\hline
\multirow{3}{*}{\textbf{Control}} 
& $\hat{S}_{SE}(X,Y)$ & 1 & 0.96298 & 0.94205 &0.90603  \\
& $\hat{S}_{CE}(X,Y)$ & 1 & 0.96938 & 0.95518 &0.87185  \\
& $\hat{S}_{E}(X,Y)$  & 1 & 0.97288 & 0.95918 &0.89940  \\
\hline
\multirow{3}{*}{\textbf{Low}} 
& $\hat{S}_{SE}(X,Y)$ &  & 1 & 0.94501 &0.93401  \\
& $\hat{S}_{CE}(X,Y)$ &  & 1&  0.95583&0.94797  \\
& $\hat{S}_{E}(X,Y)$  &  & 1&  0.91331&0.89625  \\
\hline
\multirow{3}{*}{\textbf{Medium}} 
& $\hat{S}_{SE}(X,Y)$ &  &  &  1& 0.90647 \\
& $\hat{S}_{CE}(X,Y)$ &  &  &  1& 0.93594 \\
& $\hat{S}_{E}(X,Y)$  &  &  &  1& 0.86421 \\
\hline
\multirow{3}{*}{\textbf{High}}
& $\hat{S}_{SE}(X,Y)$ &  &  &  &1  \\
& $\hat{S}_{CE}(X,Y)$ &  &  &  &1  \\
& $\hat{S}_{E}(X,Y)$  &  &  &  & 1 \\
\hline
\end{tabular}}
\end{table}
\begin{figure}
    \centering
    \includegraphics[width=0.4\linewidth]{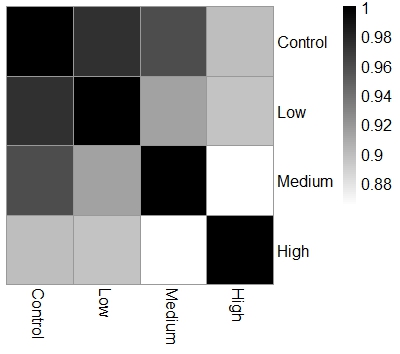}
    \caption{Extropy similarity heat map of treatment groups of mice data}
    \label{fig:R5extropysimilarityheatmap_micedata}
\end{figure}

\subsection{Application of scale invariance property of similarity measures in image analysis}

\begin{figure}[h!]
  \centering

  \begin{subfigure}{0.41\textwidth}
    \includegraphics[width=\linewidth]{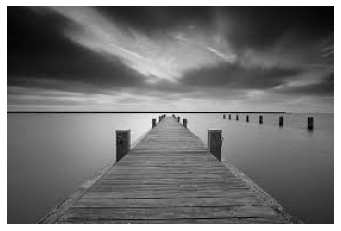}
    \caption{$\bold{A}$}
  \end{subfigure}
 \hspace{0.03\textwidth}
  \begin{subfigure}{0.41\textwidth}
    \includegraphics[width=\linewidth]{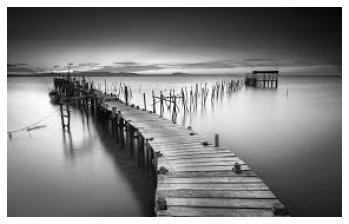}
    \caption{$\bold{B}$}
  \end{subfigure}

  \vspace{0.3cm}

  \begin{subfigure}{0.41\textwidth}
    \includegraphics[width=\linewidth]{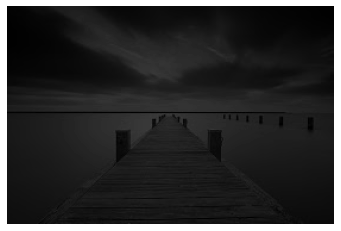}
    \caption{$\bold{A_{0.25}}$}
  \end{subfigure}
  \hspace{0.03\textwidth}
  \begin{subfigure}{0.41\textwidth}
    \includegraphics[width=\linewidth]{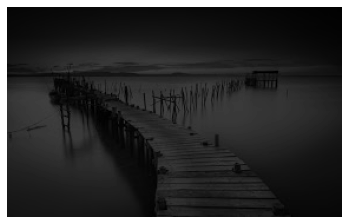}
    \caption{$\bold{B_{0.25}}$}
  \end{subfigure}

  \vspace{0.3cm}

  \begin{subfigure}{0.41\textwidth}
    \includegraphics[width=\linewidth]{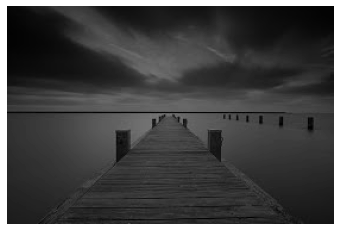}
    \caption{$\bold{A_{0.5}}$}
  \end{subfigure}
\hspace{0.03\textwidth}
  \begin{subfigure}{0.41\textwidth}
    \includegraphics[width=\linewidth]{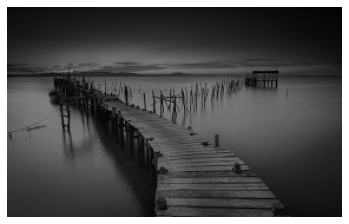}
    \caption{$\bold{B_{0.5}}$}
  \end{subfigure}
 \vspace{0.3cm}

  \begin{subfigure}{0.41\textwidth}
    \includegraphics[width=\linewidth]{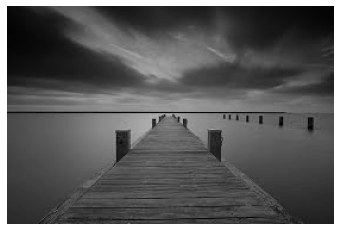}
    \caption{$\bold{A_{0.75}}$}
  \end{subfigure}
\hspace{0.03\textwidth}
  \begin{subfigure}{0.41\textwidth}
    \includegraphics[width=\linewidth]{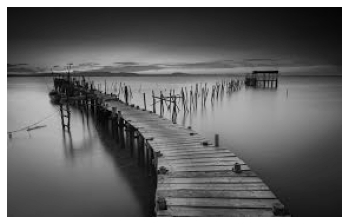}
    \caption{$\bold{B_{0.75}}$}
  \end{subfigure}

  \caption{Pair of images with different exposures}
  \label{fig:R5bridgeimages_exposures}
\end{figure}
Here, we extend the applications of defined similarity measures to analyse image data. Throughout this discussion, we represent the images using bold uppercase letters and the random variable corresponding to those images using normal uppercase letters. Note that the random variable representing an image takes pixel values of images scaled in [0,1], where 0 denotes black and 1 denotes white. 

To examine the effect of exposure variations on image similarity, we make use of the scale-invariance property of similarity ratios.  Initially, we select two images $\bold{A}$ and $\bold{B}$ in Figure \ref{fig:R5bridgeimages_exposures} and computed the three defined similarity measures between them using the grey levels (pixel values) ranges from 0 to 1. Let A and B denote the random variable that denote the grey levels of images $\bold{A}$ and $\bold{B}$ respectively. To simulate a change in exposure, we uniformly scale the pixel values of both images by a constant, $0<c\leq1$, resulting in exposure-reduced versions of the original image.  Let $\bold{X_c}$ denote the exposure-reduced image by scaling (multiplying) the grey levels of image $\bold{X}$ by $c$. Figure \ref{fig:R5bridgeimages_exposures} give the images $\bold{A}$ and $\bold{B}$ and their respective exposure-reduced counterparts $\bold{A_{0.25}}$, $\bold{A_{0.5}}$, $\bold{A_{0.75}}$, $\bold{B_{0.25}}$, $\bold{B_{0.5}}$, and $\bold{B_{0.75}}$ used for the similarity analysis. The process of exposure reduction darkens the images without altering their relative contrast structure. Now, the similarity ratios are calculated for the corresponding scaled pair of images. This approach allows us to assess whether the similarity ratios remain consistent under uniform exposure changes, thereby illustrating their scale invariance. 

 In Table \ref{Table:R5bridgeimagesimilarity}, we present the estimated survival extropy similarity $\hat{S}_{SE}(X, Y)$, cumulative extropy similarity $\hat{S}_{CE}(X, Y)$, and the extropy similarity $\hat{S}_E(X, Y)$ for the pair of images ($\bold{A_c}$, $\bold{B_c}$) for different exposure scaling values 1, 0.25, 0.5 and 0.75. 
 \begin{table}[h!]
\begin{center}
\caption{Similarity estimates for the same pair of images with scaled exposures}
\label{Table:R5bridgeimagesimilarity} 
\begin{tabular}{p{2.2cm}|p{2.2cm}|p{2cm}|p{2cm}|p{2cm}|p{2cm}}
 \hline
 Image 1($X$)& Image 2($Y$) & $c$ & $\hat{S}_{SE}(X,Y)$&$\hat{S}_{CE}(X,Y)$&$S_{E}(X,Y)$\\ 
 \hline
A&B&1&0.9815471&0.9817769&0.7891895\\
$A_{0.25}$&$B_{0.25}$&0.25&0.9815471&0.9817769&0.7891895\\
$A_{0.5}$&$B_{0.5}$&0.5& 0.9815471&0.9817769&0.7891895\\
$A_{0.75}$&$B_{0.75}$&0.75& 0.9815471&0.9817769&0.7891895\\
 \hline
\end{tabular}
\end{center}
\end{table}

Notably, all three similarity ratios remain unchanged across these varying exposures in Table \ref{Table:R5bridgeimagesimilarity}. This consistent behavior validates the scale-invariance property of the proposed extropy-based similarity ratios as well as their estimators. In particular, the similarity values are robust to linear transformations in image intensity, indicating that these measures are well-suited for comparing images taken under different lighting or exposure conditions. This is a desirable property in practical image analysis tasks, where variations in acquisition settings are common.
\begin{table}[h!]
\begin{center}
\caption{Survival extropy similarity of images with reference image}
\label{Table:R5identifybridge} 
\begin{tabular}{p{2.5cm}|p{2.5cm}|p{3cm}}
 \hline
 Black grid($Z$)& Image ($Y$)  & $\hat{S}_{SE}(Z,Y)$\\ 
 \hline
Z & A&$1.650883\times e^{-6}$\\
Z&$A_{0.25}$&$1.650883\times e^{-6}$\\
Z&$A_{0.5}$&$1.650883\times e^{-6}$\\
Z&$A_{0.75}$&$1.650883\times e^{-6}$\\
\hline
Z&B&$6.713449\times e^{-7}$\\
Z&$B_{0.25}$&$6.713449\times e^{-7}$\\
Z&$B_{0.5}$&$6.713449\times e^{-7}$\\
Z&$B_{0.75}$&$6.713449\times e^{-7}$\\
 \hline
\end{tabular}
\end{center}
\end{table}

To address the challenge of identifying the same images captured under different exposure levels or for the classification of similar images from a mixed set, we utilize a uniformly black image of the same dimensions as a reference image represented by $\bold{Z}$. As the pixel value for black is 0, the reference image results in a degenerate distribution concentrated at 0. Applying  $\hat{\bar{F}}_Z(x)=P(Z\geq x)=1 \ \ \forall x\geq0$, we  compute the survival extropy similarity between $\bold{Z}$ and each image $\bold{Y_c}$. Since $cZ=Z$, we get 
\[
S_{SE}(Z,Y)=S_{SE}(cZ,cY)=S_{SE}(Z,cY), \quad  0<c\leq 1.
\]
For instance, Table \ref{Table:R5identifybridge} presents the similarity values between the reference image and all variations of A and B, in which the values are equal for the same image with different exposures. That is, we get unique values for similarity of $\textbf{A}$ and $\textbf{B}$ with $\bold{Z}$ as
\[
S_{SE}(Z,A)=S_{SE}(cZ,cA)=S_{SE}(Z,cA)=S_A, \quad  0<c\leq 1,
\]
\[
S_{SE}(Z,B)=S_{SE}(cZ,cB)=S_{SE}(Z,cB)=S_B, \quad  0<c\leq 1.
\]
This reveals that the similarity ratio can be used to classify set of images from a mixed group of images of different exposure levels using the similarity ratio. \textbf{It is possible to derive a classification criterion for images using the similarity estimate of just a single image from each group, which can then be used to identify all other images in those groups}. The probability of true classification will always be 1. Similarly, we can extend this scenario into a large group of images as follows:

Consider a mixed group of images, and we classify them into corresponding groups. The group comprises $m$ unique images and their $k_i, i=1,2,..m$ exposure-reduced counterparts. Let $n$ denote the total number of images in the mixed set. The following is the Algorithm for the classification of $n=\sum_{i=1}^m k_i$ images into $m$ groups by identifying only one image from each group.
\begin{itemize}
    \item[\textbf{Step 1}]: Identify one image $\bold{X^{[i]}}$ from each group $i=1,2,3,... m$. Thus, we get $m$ unique images. 
    \item[\textbf{Step 2}]: Find the similarity ratio estimate  between $\bold{X^{[i]}}$ and $\bold{Z}$ and record the values. Let the similarity values corresponding to the $i^{th}$ group be $S_i$.
    \[
    S_{SE}(Z,X^{[i]})=S_i.
    \]
    Check whether $S_i$'s are unique for all $i$. 
    \item[\textbf{Step 3}]: Find the similarity of remaining $n-m$ images $\bold{A_j}, j=1,2,..,n-m$, with the reference image. Let the general notation for the similarity estimates be $S_{SE}(Z, A_j)$.
    \item[\textbf{Step 4}]: Compare all $S_{SE}(Z,A_j)$ values with the recorded values in Step:2. If $S_{SE}(Z,A_j)=S_i$ then classify the image $\bold{A_j}$ into group $i$.
\end{itemize}
It is possible to apply these classification techniques in several practical situations where exposure decreases without varying contrast, especially in fields such as image processing, photography, and machine learning, etc. Some of them are
\begin{itemize}
    \item[(i)] Security surveillance and low-light photography. In dim light, a camera might lower the exposure time or sensor gain to avoid overexposure, resulting in darker images. Example: A fixed camera capturing periodic images under the same scene; if the light dims slightly, the image gets darker, but objects remain distinguishable in terms of their relative brightness.
    \item [(ii)] Intensity normalization across scans in medical imaging (X-ray, MRI, etc.). To standardize images taken with different exposure settings or machines, pixel intensities may be scaled down to a common level. Then the brightness drops uniformly, but contrast (which is clinically more important for features like tumors, tissue boundaries) remains unaffected.
\end{itemize}


\section*{Acknowledgements}
The first author would like to thank Cochin University of Science and Technology, India, for the financial support.

\section*{Conflict of Interest statement}
On behalf of all authors, the corresponding author declares that there is no conflict of interest.

\bibliographystyle{apalike}
\bibliography{reference}
\end{document}